\documentclass[conference,letterpaper]{IEEEtran}
\IEEEoverridecommandlockouts

\usepackage{cite}
\usepackage{url}
\usepackage{ifthen}
\usepackage{algorithm, algorithmicx, algpseudocode}
\usepackage{graphicx} 
\usepackage{textcomp}
\usepackage{booktabs}
\usepackage[cmex10]{amsmath}
\usepackage{amssymb,amsfonts}
\usepackage{pifont}
\usepackage{amsthm}
\usepackage{mathrsfs}
\def\BibTeX{{\rm B\kern-.05em{\sc i\kern-.025em b}\kern-.08em T\kern-.1667em\lower.7ex\hbox{E}\kern-.125emX}}
\usepackage{tikz}
\usetikzlibrary{automata,positioning,calc}
\usetikzlibrary{intersections}
\usetikzlibrary{decorations.pathreplacing,angles,quotes}

\usepackage{bm}
\usepackage{amscd}

\algnewcommand{\Initialize}[1]{%
  \State \textbf{Initialization:}
  \Statex \hspace*{\algorithmicindent}\parbox[t]{0.8\linewidth}{\raggedright #1}
}
\makeatletter

\theoremstyle{definition}

\newtheorem{definition}{Definition}

\newtheorem{prop}{Proposition}

\newtheorem{remark}{Remark}
\newtheorem{eg}{Example}

\newcommand{\abs}[1]{\left\lvert#1\right\rvert}

\newcommand{\one}[1]{\mbox{1}\hspace{-0.25em}\mbox{l}_{\left\{#1\right\}}}
\newcommand{\argmax}{\operatornamewithlimits{argmax}}
\newcommand{\argmin}{\operatornamewithlimits{argmin}}

\def\vE{\mathbb E}

\def\vV{\mathbb V}

\font\b=cmr10 scaled\magstep4

\def\bigzerou{\smash{\lower1.7ex\hbox{\b 0}}}
\def\bigzerou{\smash{\lower1.7ex\hbox{\b 0}}}

\begin{document}
\title{A Generalized Information Bottleneck Method: A Decision-Theoretic Perspective \\
\thanks{This work was supported by JSPS KAKENHI Grant Number JP23K16886 and JST CRONOS Grant Number JPMJCS25N5.}
}

\author{
\IEEEauthorblockN{Akira Kamatsuka}
\IEEEauthorblockA{Shonan Institute of Technology \\ 
Email: \text{kamatsuka@info.shonan-it.ac.jp}
 }
\and 
\IEEEauthorblockN{Takahiro Yoshida}
\IEEEauthorblockA{Nihon University \\ 
Email: \text{yoshida.takahiro@nihon-u.ac.jp}
 }
}
\maketitle
\begin{abstract}
The information bottleneck (IB) method seeks a compressed representation of data that preserves information relevant to a target variable for prediction while discarding irrelevant information from the original data.
In its classical formulation, the IB method employs mutual information to evaluate the compression between the original and compressed data and the utility of the representation for the target variable. 
In this study, we investigate a generalized IB problem, where the evaluation of utility is based on the $\mathcal{H}$-mutual information that satisfies the concave (\texttt{CV}) and averaging (\texttt{AVG}) conditions.
This class of information measures admits a statistical decision-theoretic interpretation via its equivalence to the expected value of sample information.
Based on this interpretation, we derive an alternating optimization algorithm to assess the tradeoff between compression and utility in the generalized IB problem.
\end{abstract}

\section{Introduction} \label{sec:intro}

The information bottleneck (IB) method, which was originally proposed by Tishby \textit{et al.} \cite{tishby99information}, provides a principled framework for extracting a compressed representation $T$ from data $X$ that is relevant to a target variable $Y$. 
Its objective is to preserve as much information about $Y$ as possible while discarding information in $X$ that is irrelevant for prediction.
Since its inception, the IB framework has been extensively studied and employed in a wide range of problems in machine learning, including clustering \cite{10.1145/345508.345578,NIPS1999_be3e9d3f,6788460} and deep learning \cite{DBLP:journals/corr/AlemiFD016,7133169,shwartzziv2017openingblackboxdeep} (see \cite{2024IBSurvey} for a comprehensive survey).
Additionally, from an information-theoretic perspective, it is closely related to coding problems such as remote source coding \cite{1057738} and chief executive officer source coding under logarithmic loss \cite{6651793}, the Wyner--Ahlswede--K\"orner problem \cite{1055374,1055469}, and a hypothesis testing against independence \cite{1057194}.

A central feature of the IB framework is the tradeoff between compression and predictive performance.
In the original formulation \cite{tishby99information}, this tradeoff is characterized by the Shannon mutual information (MI) $I(X;T)$ as a measure of compression and $I(Y;T)$ as a measure of predictive utility.
An extended Arimoto--Blahut algorithm was initially proposed to determine the corresponding tradeoff curve.
Subsequently, several generalized IB formulations were introduced, where more general measures replaced these information measures.
Hsu \textit{et al.} \cite{8437632} and Shahab and Flavio \cite{e22111325} considered formulations based on the Arimoto MI \cite{arimoto1977} and the $f$-information \cite{1055015}. 
Mimura \cite{mimura_shannon_en} proposed a formulation based on the $f$-information for both the compression and utility.

Despite these developments, the operational meaning of adopting a particular generalized utility measure in the IB framework remains unclear.
In practical applications, the compressed representation $T$ of data is typically used by a decision-maker to select an action; therefore, its usefulness should be evaluated in terms of decision performance rather than solely in terms of abstract information quantities.

Based on this observation, we adopt a class of information measures known as $\mathcal{H}$-mutual information ($\mathcal{H}$-MI) \cite{9505206} that serves as the utility measure in the IB framework.
This class includes information quantities of the form 
$I_{\mathcal{H}}(Y;T)=\mathcal{H}(Y)-\mathcal{H}(Y|T)$,
where $\mathcal{H}$ is an entropy-like functional.
We focus on the subclass that satisfies the concave (\texttt{CV}) and averaging (\texttt{AVG}) conditions introduced by Alvim \textit{et al.} \cite{ALVIM201932}; these conditions ensure fundamental properties such as nonnegativity and data-processing inequality (DPI).
Importantly, this subclass coincides with the \textit{expected value of sample information} (EVSI) in statistical decision theory \cite{raiffa1961applied}.
Specifically, each $\mathcal{H}$-MI admits a representation as an EVSI associated with an appropriate decision problem \cite{AkiraKAMATSUKA20252024TAP0010}.
Therefore, adopting $\mathcal{H}$-MI as the utility measure provides the IB framework with an explicit decision-theoretic interpretation: the usefulness of $T$ is evaluated in terms of a decision rule $\delta\colon \mathcal{T}\to\mathcal{A}$, the resulting action $A=\delta(T)$, and a loss function $\ell(y,a)$, as shown in Figure \ref{fig:system_model_action_IBM}.

However, although this interpretation clarifies the meaning of utility, it does not provide a systematic method for evaluating the tradeoff between compression and utility. 
In particular, the existing generalized IB formulations do not explicitly exploit the decision-theoretic structure of $\mathcal{H}$-MI in the design of optimization algorithms, 
leading to the following question: how can the decision-theoretic interpretation induced by $\mathcal{H}$-MI be incorporated directly into both the formulation and the algorithmic design of the IB problem?

In this study, we attempt to answer this question.
The main contributions of our work are summarized as follows:
\begin{itemize}
\item We formulate a generalized IB problem, where the utility of the compressed representation is measured by $\mathcal{H}$-MI that satisfies the \texttt{CV} and \texttt{AVG} conditions, and the compression measure remains the Shannon MI.
\item By exploiting the equivalence between this class of $\mathcal{H}$-MI and the EVSI, we derive an alternating optimization algorithm to determine the compression--utility trade-off curve.
\end{itemize}

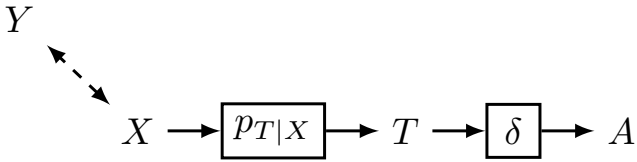
\begin{figure}[tb]
\centering
\resizebox{0.5\textwidth}{!}{
\begin{tikzpicture}[auto]
\tikzset{block/.style={draw, rectangle, minimum height = .5cm, minimum width = .5cm, text centered, thick}};
\node (database) {$X$}; 
\node[above left=.8cm of database] (label) {$Y$};
\node[block,right=.5cm of database] (mechanism) {${p_{T\mid X}}$};
\node[right=.5cm of mechanism] (output) {$T$};
\node[block, right=.5cm of output] (d_bob) {$\delta$};
\node[right=.5cm of d_bob] (T_hat) {$A$};

\draw[<->, thick, >=latex, dashed] (database) --  (label);
\draw[->, thick, >=latex] (database) --  (mechanism);
\draw[->, thick, >=latex] (mechanism) -- (output);
\draw[->, thick, >=latex] (output) -- (d_bob);
\draw[->, thick, >=latex] (d_bob) -- (T_hat);
\end{tikzpicture}
}
\caption{A system model for the information bottleneck problem incorporating a decision-maker.}
\label{fig:system_model_action_IBM}
\end{figure}

\section{Preliminaries}\label{sec:preliminary}
Let $Y$ and $X$ be random variables representing the prediction target and associated data, respectively, where $(X,Y)\sim p_{X,Y}$; the joint distribution $p_{X,Y}$ is assumed to be known.
Let $T$ denote a compressed representation and $p_{T\mid X}$ be a stochastic mechanism that maps $X$ to $T$.
We assume that $Y\to X\to T$ form a Markov chain in this order.
The alphabets of $X$, $Y$, and $T$, which are denoted as $\mathcal{X}$, $\mathcal{Y}$, and $\mathcal{T}$, respectively, are assumed to be finite.
Let $H(Y):=-\sum_{y}p_Y(y)\log p_Y(y)$ and $H(Y | T):=-\sum_{t}p_T(t)\sum_{y}p_{Y \mid T}(y | t)\log p_{Y\mid T}(y | t)$ denote
the entropy of $Y$ and the conditional entropy of $Y$ given $T$, respectively. We define the mutual information between $Y$ and $T$ as
$I(Y;T):=H(Y)-H(Y | T)$.
Furthermore, for probability distributions $p$ and $q$, the Kullback--Leibler divergence is defined as 
$D(p || q):=\sum_{z}p(z)\log\frac{p(z)}{q(z)}$.
Throughout this study, $\log$ denotes the natural logarithm.
Let $\mathcal{A}$ denote a decision space and $\delta\colon\mathcal{T}\to\mathcal{A}$ be a decision rule of a decision-maker.
We define the random variable $A:=\delta(T)$ as the action taken by the decision-maker.
Let $\ell(y, a)$ be a loss function. We define the Bayes risk of $\delta$ as
$r(\delta):=\vE_{Y,T}[\ell(Y,\delta(T))]$, where $\vE_{Y, T}[\cdot]$ denotes the expectation under $p_{Y, T}$.
The set of all probability distributions on $\mathcal{Y}$, which is identified with the $(|\mathcal{Y}|-1)$-dimensional probability simplex, is denoted as $\Delta_{\mathcal{Y}}$.

In this section, we initially review a fundamental property of Shannon MI.
Then, we introduce EVSI, which forms a class of information measures in statistical decision theory. 
Finally, we discuss $\mathcal{H}$-MI, which is a broad class of information measures including both Shannon MI and EVSI.

\subsection{A Property of Shannon MI and EVSI}

Shannon MI $I(X; T)$ admits the following variational characterization regarding the Kullback--Leibler divergence.

\begin{prop}[\text{\cite[Lemma 10.8.1]{Cover:2006:EIT:1146355}}]\label{prop:vc_Shannon_MI}
Let $(X,T)\sim p_{X,T}$. Then,
\begin{align}
I(X;T)
&= \min_{q_T} D(p_X p_{T\mid X}\|p_X q_T),
\end{align}
where the minimum is achieved when $q_T=p_T$.
\end{prop}

Before introducing EVSI, we review a basic result obtained from statistical decision theory.

\begin{prop}[\text{\cite[Thm. 2.7]{ghosh2007introduction}}]
\begin{align}
\min_{\delta} r(\delta)
&= \vE_T\!\left[\min_{a}\vE_Y[\ell(Y,a)\mid T]\right]. \label{eq:sdt_basic}
\end{align}
The decision rule that achieves the minimum in \eqref{eq:sdt_basic}, which is referred to as the \emph{Bayes decision rule}, is given as 
\begin{align}
\delta^{*}(t)
&= \argmin_{a}\vE_Y[\ell(Y,a)\mid T=t]. \label{eq:Bayes_decision_rule}
\end{align}
\end{prop}

\begin{eg}[Point estimation of $Y$]
Let us consider the problem of point estimation of $Y$, which is denoted as $A=\hat{Y}$, under the squared-error loss
$\ell_{\mathrm{sq}}(y,\hat{y})=(y-\hat{y})^{2}$.
Then,
\begin{align}
\min_{\delta} r(\delta)
&= \vE_{T}[\vV(Y\mid T)],
\end{align}
where $\vV(Y | T)$ denotes the conditional variance; the Bayes decision rule is given as 
$\delta^{*}(t)=\vE_Y[Y | T=t]$.
Under the $0$--$1$ loss
$\ell_{0\text{-}1}(y,\hat{y})=1-\one{\{y\}}(\hat{y})$,
we obtain
\begin{align}
\min_{\delta} r(\delta)
&= 1-\vE_T\!\left[\max_{y}p_{Y\mid T}(y\mid T)\right],
\end{align}
and the Bayes decision rule is given as 
$\delta^{*}(t)=\argmax_{y}p_{Y\mid T}(y\mid t)$, i.e., the maximum a posteriori estimator.
\end{eg}

\begin{eg}[Decision on the probability distribution of $Y$]
Let us consider the problem of deciding a probability distribution of $Y$, which is denoted as $A=r\in\Delta_{\mathcal{Y}}$.
Under the logarithmic loss $\ell_{\mathrm{log}}(y,r):=-\log r(y)$, we obtain 
\begin{align}
\min_{\delta} r(\delta)
&= H(Y\mid T),
\end{align}
and the Bayes decision rule is given as $\delta^{*}(t)=p_{Y\mid T}(\cdot\mid t)$.
Under the $\alpha$-loss \cite{8804205} $\ell_{\alpha}(y,r):=\frac{\alpha}{\alpha-1}\left(1-r(y)^{1-\frac{1}{\alpha}}\right)$, the optimal Bayes risk is given as
\begin{align}
\min_{\delta} r(\delta)
&= \frac{\alpha}{\alpha-1}
\left(
1-\exp\left\{\frac{1-\alpha}{\alpha}H_{\alpha}^{\mathrm{A}}(Y\mid T)\right\}
\right),
\end{align}
where
\begin{align}
H_{\alpha}^{\mathrm{A}}(Y\mid T)
:=\frac{\alpha}{1-\alpha}\log
\sum_{t}p_T(t)\!\left(\sum_{y}p_{Y\mid T}(y\mid t)^{\alpha}\right)^{\frac{1}{\alpha}}
\end{align}
denotes the Arimoto conditional entropy \cite{arimoto1977}; the optimal decision rule is given as 
$\delta^{*}(t)=\frac{p_{Y\mid T}(\cdot\mid t)^{\alpha}}{\sum_{y}p_{Y\mid T}(y\mid t)^{\alpha}}$ \cite[Lemma 1]{8804205}.
\end{eg}

In the problem of deciding the probability distribution of $Y$, i.e., $\mathcal{A}=\Delta_{\mathcal{Y}}$, the loss function $\ell(y,r)$ is often referred to as a \emph{scoring rule} \cite{doi:10.1198/016214506000001437,Dawid:2014ua}.
Among the scoring rules, those that satisfy the following condition are referred to as \emph{proper scoring rules} (PSRs).

\begin{definition}[Proper scoring rule]
Let $\mathcal{A} = \Delta_{\mathcal{Y}}$. 
A scoring rule $\ell(y,r)$ is a PSR if, for all $r\in\Delta_{\mathcal{Y}}$,
\begin{align}
\vE_Y[\ell(Y,p_Y)]
\leq \vE_Y[\ell(Y,r)]. \label{eq:PSR}
\end{align}
\end{definition}

\begin{eg}
$\ell_{\mathrm{log}}(y,r)$ is a PSR, whereas $\alpha$-loss $\ell_{\alpha}(y,r)$ is not.
\end{eg}

\begin{remark}\label{rem:optimal_dist_psr}
For a decision problem on the probability distribution of $Y$ with a PSR $\ell(y,r)$, the Bayes decision rule is given by
\begin{align}
\text{EVSI}^{\ell}(Y;T)\leq \text{EVPI}^{\ell}(Y).
\end{align}
\end{remark}

Next, we introduce EVSI \cite{raiffa1961applied}.
EVSI is defined as the difference between the optimal Bayes risk with access to $T$ and that with no access to $T$\footnote{In \cite{raiffa1961applied}, EVSI was originally defined in terms of a utility function $u(y, a)$ rather than a loss function $\ell(y, a)$.}.

\begin{definition}[EVSI {\cite[Eq.(4-35)]{raiffa1961applied}}]
Let $(Y,T)\sim p_{Y,T}=p_Y p_{T\mid Y}$ and $\ell(y,a)$ be a loss function.
The EVSI of $T$ on $Y$ is defined as
\begin{align}
&\text{EVSI}^{\ell}(Y;T)\notag \\ 
&:= \min_{a}\vE_Y[\ell(Y,a)]
   - \min_{\delta}\vE_{Y,T}[\ell(Y,\delta(T))] \notag\\
&= \min_{a}\vE_Y[\ell(Y,a)]
   - \vE_T\!\left[\min_{a}\vE_Y[\ell(Y,a)\mid T]\right].
\end{align}
\end{definition}

\begin{eg}
For $A=\hat{Y}$ and $\ell_{\mathrm{sq}}(y,\hat{y})=(y-\hat{y})^{2}$,
\begin{align}
\text{EVSI}^{\ell_{\mathrm{sq}}}(Y;T)
= \vV(Y)-\vE_T[\vV(Y\mid T)].
\end{align}
\end{eg}

\begin{eg}
For $A=r\in\Delta_{\mathcal{Y}}$ and $\ell_{\mathrm{log}}(y,r)=-\log r(y)$,
\begin{align}
\text{EVSI}^{\ell_{\mathrm{log}}}(Y;T)
= H(Y)-H(Y\mid T)
= I(Y;T).
\end{align}
\end{eg}

The EVSI is upper-bounded by the expected value of perfect information (EVPI), which is defined below.

\begin{definition}[EVPI {\cite[Eq.(5-8)]{raiffa1961applied}}]
The EVPI of $Y$ is defined as
\begin{align}
\text{EVPI}^{\ell}(Y)
:= \min_{a}\vE_Y[\ell(Y,a)]
 - \vE_Y[\min_{a}\ell(Y,a)].
\end{align}
\end{definition}

\begin{prop}\label{prop:evsi_leq_evpi}
\begin{align}
\text{EVSI}^{\ell}(Y;T)\leq \text{EVPI}^{\ell}(Y).
\end{align}
\end{prop}

\begin{proof}
This follows directly from the definitions.
\end{proof}

\begin{remark}
Let us assume that the loss function $\ell(y,a)$ is nonnegative and satisfies the following condition:
for every $y\in \mathcal{Y}$, there exists $a\in \mathcal{A}$ such that $\ell(y,a)=0$.
Then,
\begin{align}
\text{EVPI}^{\ell}(Y)
= \min_{a}\vE_Y[\ell(Y,a)].
\end{align}
\end{remark}

\begin{eg}
Let $\mathcal{A}=\Delta_{\mathcal{Y}}$.
Considering that the logarithmic loss $\ell_{\mathrm{log}}(y,r)=-\log r(y)$ is nonnegative and satisfies the above condition, we obtain
\begin{align}
\text{EVPI}^{\ell_{\mathrm{log}}}(Y)
= H(Y).
\end{align}
\end{eg}

\subsection{Definition and Properties of $\mathcal{H}$-MI}
Next, we introduce the $\mathcal{H}$-MI and the conditions under which it satisfies the nonnegativity and DPI.

\begin{definition}[$\mathcal{H}$-MI {\cite[Eq.(5)]{9505206}}]
Let $\mathcal{H}(Y)=\mathcal{H}(p_Y)$ be a continuous functional of $p_Y$ and $\mathcal{H}(Y | T)=\mathcal{H}(p_Y,p_{T\mid Y})$ be a functional of $(p_Y,p_{T\mid Y})$.
$\mathcal{H}$-MI is defined as
\begin{align}
I_{\mathcal{H}}(Y;T) &:= \mathcal{H}(Y)-\mathcal{H}(Y\mid T).
\end{align}
\end{definition}

\begin{definition}[\texttt{CV} and \texttt{AVG} conditions]
Functional $\mathcal{H}(Y)$ satisfies the \texttt{CV} condition if it is concave in $p_Y$.
Functional $\mathcal{H}(Y | T)$ satisfies the \texttt{AVG} condition if it can be expressed as
\begin{align}
\mathcal{H}(Y\mid T)
&= \vE_T[\mathcal{H}(p_{Y\mid T}(\cdot\mid T))].
\end{align}
\end{definition}

Alvim \textit{et al.} \cite{ALVIM201932} showed that, under the \texttt{CV} condition, the \texttt{AVG} condition is equivalent to both the nonnegativity of $I_{\mathcal{H}}(Y;T)$ and the DPI.
\begin{prop}[\text{\cite[Props. 14, 15, 16, and 18]{ALVIM201932}}]
For $\mathcal{H}$-MI $I_{\mathcal{H}}(Y;T)=\mathcal{H}(Y)-\mathcal{H}(Y\mid T)$ with $\mathcal{H}(Y)$ satisfying the \texttt{CV} condition, the following are equivalent:
\begin{enumerate}
\item $\mathcal{H}(Y | T)$ satisfies the \texttt{AVG} condition;
\item $I_{\mathcal{H}}(Y;T)\geq 0$ (nonnegativity);
\item $I_{\mathcal{H}}(Y;T)$ satisfies the DPI, i.e., for any Markov chain $Y\to T\to W$, 
\begin{align}
I_{\mathcal{H}}(Y;T)\leq I_{\mathcal{H}}(Y;W). 
\end{align}
\end{enumerate}
\end{prop}

\begin{remark}
Americo \textit{et al.} \cite{9064819} extended this result using the core-concave (\texttt{CCV}) and $\eta$-averaging (\texttt{EAVG}) conditions.
Note that their results were stated in terms of conditions on $\mathcal{H}(Y | T)$ rather than on $\mathcal{H}$-MI itself.
\end{remark}

\begin{eg}
Both the Shannon MI $I(Y;T)$ and EVSI induced by a given statistical decision problem ($\mathrm{EVSI}^{\ell}(Y;T)$) are $\mathcal{H}$-MIs that satisfy the \texttt{CV} and \texttt{AVG} conditions (see \cite[Sec. V.F]{9064819}).
\end{eg}

Conversely, Kamatsuka \textit{et al.} \cite[Thm. 2]{AkiraKAMATSUKA20252024TAP0010} proved that any $\mathcal{H}$-MI that satisfies the \texttt{CV} and \texttt{AVG} conditions can be represented as an EVSI for an appropriate statistical decision problem of deciding a probability distribution on $\mathcal{Y}$, i.e., $\mathcal{A}=\Delta_{\mathcal{Y}}$.
This is because any concave function on $\Delta_{\mathcal{Y}}$ can be represented as the minimum expected loss of a suitable decision problem of deciding a probability distribution \cite[Sec. 3.5.4]{grunwald2004}\footnote{Technically, this result has been proved under the \texttt{CCV} and \texttt{EAVG} conditions.}.
Consequently, $\mathcal{H}$-MI admits the variational representation regarding the reverse channel $r_{Y\mid T}$ presented below.

\begin{prop}[\text{\cite[Thm. 2]{AkiraKAMATSUKA20252024TAP0010}}]\label{prop:vc_H_MI}
Let $\mathcal{A}=\Delta_{\mathcal{Y}}$ and $I_{\mathcal{H}}(Y;T)=\mathcal{H}(Y)-\mathcal{H}(Y\mid T)$ be an $\mathcal{H}$-MI that satisfies the \texttt{CV} and \texttt{AVG} conditions.
Then, there exists a PSR $\ell_{\mathcal{H}}(y,r)$ such that
\begin{align}
I_{\mathcal{H}}(Y;T)
&= \text{EVSI}^{\ell_{\mathcal{H}}}(Y;T) \notag\\
&= \max_{r_{Y\mid T}} \mathcal{F}_{\mathcal{H}}(p_{T\mid Y}, r_{Y\mid T}),
\end{align}
where
\begin{align}
\mathcal{F}_{\mathcal{H}}(p_{T\mid Y}, r_{Y\mid T})
:= \mathcal{H}(Y)
 - \vE_{Y,T}[\ell_{\mathcal{H}}(Y,r_{Y\mid T}(\cdot\mid T))].
\label{eq:vc_H_MI}
\end{align}
\end{prop}

\begin{remark}
Let $\mathcal{Y}=\{y_{1},\dots,y_{m}\}$ and let $\mathcal{H}(p_Y)$ be a continuous concave functional.
Based on \cite[Sec. 3.5.4]{grunwald2004}, a possible choice of $\ell_{\mathcal{H}}(y,r)$ in Proposition \ref{prop:vc_H_MI} is
\begin{align}
\ell_{\mathcal{H}}(y,r)
:= \mathcal{H}(r) + z^{\top}(\mbox{1}\hspace{-0.25em}\mbox{l}^{y}-r),
\label{eq:l_H}
\end{align}
where $\mbox{1}\hspace{-0.25em}\mbox{l}^{y}$ is an $m$-dimensional vector, where its $y$-th component is $1$ and all its other components are $0$; 
$z\in\partial\mathcal{H}(r)$ is a subgradient.
If $\mathcal{H}(p_{Y})$ is differentiable, then $\partial\mathcal{H}(r)=\{\nabla\mathcal{H}(r)\}$.
\end{remark}

\begin{eg}\label{eg:entropy}
If $\mathcal{H}(p_Y)=H(p_Y)=-\sum_{y}p_Y(y)\log p_Y(y)$ and
$\ell_{\mathcal{H}}(y,r)=\ell_{\mathrm{log}}(y,r)$,
then
\begin{align}
I(Y;T)
&= \text{EVSI}^{\ell_{\mathcal{H}}}(Y;T) \\
&= \max_{r_{Y\mid T}}
   \vE_{Y,T}\!\left[\log\frac{r_{Y\mid T}(Y\mid T)}{p_Y(Y)}\right].
\end{align}
\end{eg}

\begin{eg}\label{eg:variance}
If $\mathcal{H}(p_Y)=\vV(Y)=\vE[(Y-\vE[Y])^{2}]$ and
$\ell_{\mathcal{H}}(y,r)=\left(y-\sum_{y}r(y)y\right)^{2}$, then
\begin{align}
&\text{EVSI}^{\ell_{\mathrm{sq}}}(Y;T) = \text{EVSI}^{\ell_{\mathcal{H}}}(Y;T) \\ 
&= \max_{r_{Y\mid T}}
   \left\{
   \vV(Y)
   - \vE_{Y,T}[\ell_{\mathcal{H}}(Y,r_{Y\mid T}(\cdot\mid T))]
   \right\}.
\end{align}
\end{eg}

\section{A Generalized Information Bottleneck Problem} \label{sec:SDT_IBM}

In this section, we formulate a generalized IB problem, where the Shannon MI $I(X;T)$ is adopted as the compression measure and the utility is measured using an $\mathcal{H}$-MI 
that satisfies the \texttt{CV} and \texttt{AVG} conditions.
Throughout this section, we denote $I(X;T)=I(p_X,p_{T\mid X})$ and $I_{\mathcal{H}}(Y;T)=I_{\mathcal{H}}(p_Y,p_{T\mid Y})$, where $p_{T\mid Y}$ is induced by $p_{T\mid X}$ as follows:
\begin{align}
p_{T\mid Y}(t\mid y)
&:=\sum_{x}p_{T\mid X}(t\mid x)\,p_{X\mid Y}(x\mid y).
\end{align}

\subsection{Formulation of the Generalized IB problem}

The tradeoff between compression and utility can be formulated as the constrained optimization problem described below.
\begin{definition}
\begin{align}
\text{IB}_{\mathcal{H}}(R)
:= \min_{p_{T\mid X}\,:\, I_{\mathcal{H}}(p_Y,p_{T\mid Y})\ge R}
I(p_X,p_{T\mid X}).
\label{eq:generalized_IB}
\end{align}
\end{definition}

By introducing a Lagrange multiplier $\beta\in[0,\infty]$, the problem of finding a minimizer of \eqref{eq:generalized_IB} can be reduced to the following unconstrained optimization problem:
\begin{align}
\min_{p_{T\mid X}}
\left\{
I(p_X,p_{T\mid X})
-\beta\, I_{\mathcal{H}}(p_Y,p_{T\mid Y})
\right\}.
\label{eq:generalized_IB_optimization}
\end{align}

\begin{remark}
When $I_{\mathcal{H}}(p_Y,p_{T\mid Y})=I(p_Y,p_{T\mid Y})$ (Shannon MI), \eqref{eq:generalized_IB} reduces to the original IB problem \cite{tishby99information}.
\end{remark}

\subsection{Derivation of an Alternating Optimization Algorithm}

We derive an alternating optimization algorithm to solve \eqref{eq:generalized_IB_optimization}.
Using the results in Section \ref{sec:preliminary}, \eqref{eq:generalized_IB_optimization} is reduced to a triple minimization problem as follows:

\begin{align}
&\min_{p_{T\mid X}}
\left\{
I(p_X,p_{T\mid X})
-\beta\, I_{\mathcal{H}}(p_Y,p_{T\mid Y})
\right\} \notag\\
&\overset{(a)}{=}
\min_{p_{T\mid X}}
\left\{
\min_{q_T} D(p_X p_{T\mid X}\|p_X q_T)
-\beta \max_{r_{Y\mid T}}\mathcal{F}_{\mathcal{H}}(p_{T\mid Y},r_{Y\mid T})
\right\} \notag\\
&=
\min_{p_{T\mid X}}\min_{q_T}\min_{r_{Y\mid T}}
\left\{
D(p_X p_{T\mid X}\|p_X q_T)
-\beta\, \mathcal{F}_{\mathcal{H}}(p_{T\mid Y},r_{Y\mid T})
\right\},
\end{align}

where $(a)$ follows from Propositions \ref{prop:vc_Shannon_MI} and \ref{prop:vc_H_MI}.
Fixing two of the three distributions and minimizing over the remaining one yields update rules described below.

\begin{prop}\label{prop:update_formulae_SDT_IBM}
Let $\beta\in[0,\infty]$ be a Lagrange multiplier and define
\begin{align}
&G_{\beta}(p_{T\mid X},q_T,r_{Y\mid T}) \notag \\ 
&:= D(p_X p_{T\mid X}\|p_X q_T)
-\beta\, \mathcal{F}_{\mathcal{H}}(p_{T\mid Y},r_{Y\mid T}).
\end{align}
Then:
\begin{enumerate}
\item
For a fixed $(p_{T\mid X},q_T)$, $G_{\beta}(p_{T\mid X},q_T,r_{Y\mid T})$ is minimized using 

\begin{align}
r_{Y\mid T}^{*}(y\mid t)
&= p_{Y\mid T}(y\mid t) \notag\\
&= \sum_{x} p_{Y\mid X}(y\mid x)\cdot \frac{p_X(x)p_{T\mid X}(t\mid x)}{\sum_{x} p_X(x)p_{T\mid X}(t\mid x)}. \label{eq:update_r_Y_T}
\end{align}

\item
For a fixed $(p_{T\mid X},r_{Y\mid T})$, $G_{\beta}(p_{T\mid X},q_T,r_{Y\mid T})$ is minimized using 
\begin{align}
q_T^{*}(t)
&= p_T(t)
= \sum_{x} p_X(x)\,p_{T\mid X}(t\mid x). \label{eq:update_q_T}
\end{align}
\item
For a fixed $(q_T,r_{Y\mid T})$, $G_{\beta}(p_{T\mid X},q_T,r_{Y\mid T})$ is minimized using 

\begin{align}
&p_{T\mid X}^{*}(t\mid x) \notag \\ 
&=
\frac{q_T(t)\exp\left\{-\beta \sum_{y}p_{Y\mid X}(y | x)\ell_{\mathcal{H}}\bigl(y,r_{Y\mid T}(\cdot | t)\bigr)\right\}}
{\sum_{t} q_T(t)\exp\left\{-\beta \sum_{y}p_{Y\mid X}(y | x)\ell_{\mathcal{H}}\bigl(y,r_{Y\mid T}(\cdot | t)\bigr)\right\}}.
\label{eq:update_p_T_X}
\end{align}
\end{enumerate}
\end{prop}

\begin{proof}
See Appendix \ref{proof:update_formulae_SDT_IBM}.
\end{proof}

Based on Proposition \ref{prop:update_formulae_SDT_IBM}, we derive an alternating optimization algorithm (Algorithm \ref{alg:AO_SDT_IBM}) to calculate
$\min_{p_{T\mid X}}\min_{q_T}\min_{r_{Y\mid T}}G_{\beta}(p_{T\mid X},q_T,r_{Y\mid T})$, 
where $p_{T\mid X}^{(0)}$ denotes an initial distribution.

\begin{algorithm}[h]
	\caption{Alternating optimization algorithm for solving the generalized IB problem}
	\label{alg:AO_SDT_IBM}
	\begin{algorithmic}[1]
		\Require 
			\Statex $p_{X, Y}, \beta \in [0, \infty]$, $\epsilon\in (0, 1), p_{T\mid X}^{(0)}$
		\Ensure
			\Statex Approximate value of \\ $\min_{p_{T\mid X}}\min_{q_{T}}\min_{r_{Y\mid T}}G_{\beta}(p_{T\mid X}, q_{T}, r_{Y\mid T})$
		\Initialize{
			$q_{T}^{(0)}(t) \gets \sum_{x}p_{X}(x)p_{T\mid X}^{(0)}(t\mid x)$ \\ 
			$r_{Y\mid T}^{(0)} \gets \sum_{x}p_{Y\mid X}(y\mid x) \cdot \frac{p_{X}(x)p_{T\mid X}^{(0)}(t\mid x)}{q_{T}^{(0)}(t)}$ \\ 
			$G^{(0)}\gets G_{\beta}(p_{T\mid X}^{(0)}, q_{T}^{(0)}, r_{Y\mid T}^{(0)})$ \\ 
			$k\gets 0$ \\
      }
		\Repeat
			\State $k\gets k+1$
			\State $p_{T\mid X}^{(k)}(t | x) \gets \frac{q_{T}^{(k)}(t)\exp\{-\beta \sum_{y}p_{Y\mid X}(y\mid x)\ell_{\mathcal{H}}(y, r_{Y\mid T}^{(k)}(\cdot \mid t))\}}{\sum_{t}q_{T}^{(k)}(t)\exp\{-\beta \sum_{y}p_{Y\mid X}(y\mid x)\ell_{\mathcal{H}}(y, r_{Y\mid T}^{(k)}(\cdot \mid t))\}}$ 
			\State $q_{T}^{(k)}(t) \gets \sum_{x}p_{X}(x)p_{T\mid X}^{(k)}(t\mid x)$  
			\State $r_{Y\mid T}^{(k)} \gets \sum_{x}p_{Y\mid X}(y\mid x) \cdot \frac{p_{X}(x)p_{T\mid X}^{(k)}(t\mid x)}{q_{T}^{(k)}(t)}$  
			\State $G^{(k)}\gets G_{\beta}(p_{T\mid X}^{(k)}, q_{T}^{(k)}, r_{Y\mid T}^{(k)})$
		\Until{$\abs{G^{(k)} - G^{(k-1)}} < \epsilon$} \Comment{stopping condition}
		\State \textbf{return} $G^{(k)}$
	\end{algorithmic}
\end{algorithm}

\begin{remark}
When $I_{\mathcal{H}}(Y;T)=I(Y;T)=\mathrm{EVSI}^{\ell_{\mathrm{log}}}(Y;T)$,
the update rules for $r_{Y\mid T}$ and $q_T$ remain as described in \eqref{eq:update_r_Y_T} and \eqref{eq:update_q_T}, respectively.
The update rule for $p_{T\mid X}$ becomes
\begin{align}
&p_{T\mid X}^{*}(t\mid x) \notag \\
&=
\frac{
q_T(t)\exp\left\{-\beta \sum_y p_{Y\mid X}(y\mid x)\,
\ell_{\mathrm{log}}\bigl(y,r_{Y\mid T}(\cdot\mid t)\bigr)\right\}
}{
\sum_{t} q_T(t)\exp\left\{-\beta \sum_y p_{Y\mid X}(y\mid x)\,
\ell_{\mathrm{log}}\bigl(y,r_{Y\mid T}(\cdot\mid t)\bigr)\right\}
} \notag\\
&=
\frac{
q_T(t)\exp\left\{\beta \sum_y p_{Y\mid X}(y\mid x)\log r_{Y\mid T}(y\mid t)\right\}
}{
\sum_{t} q_T(t)\exp\left\{\beta \sum_y p_{Y\mid X}(y\mid x)\log r_{Y\mid T}(y\mid t)\right\}
} \notag\\
&=
\frac{
q_T(t)\exp\left\{-\beta D\left(p_{Y\mid X}(\cdot\mid x)\,\|\,r_{Y\mid T}(\cdot\mid t)\right)\right\}
}{
\sum_{t} q_T(t)\exp\left\{-\beta D\left(p_{Y\mid X}(\cdot\mid x) || r_{Y\mid T}(\cdot\mid t)\right)\right\}
}.
\end{align}
This coincides with the update rule in the original IB problem \cite{tishby99information}.
\end{remark}

\subsection{Discussion on the Convergence of the Alternating Optimization Algorithm}
\label{ssec:convergence}

In this section, we discuss the convergence of Algorithm \ref{alg:AO_SDT_IBM}.
Let $\{p_{T\mid X}^{(k)}\}_{k=0}^{\infty}$, $\{q_T^{(k)}\}_{k=0}^{\infty}$, and $\{r_{Y\mid T}^{(k)}\}_{k=0}^{\infty}$ denote the sequences generated by Algorithm \ref{alg:AO_SDT_IBM}.

In general, the objective function of the generalized IB problem is nonconvex, similar to the original IB problem; therefore, the global convergence of Algorithm \ref{alg:AO_SDT_IBM} cannot be ensured.
However, the following results establish a local convergence regarding the objective function values. 

\begin{prop}\label{prop:monotonically_decreasing}
Sequence $\{G_{\beta}(p_{T\mid X}^{(k+1)},q_T^{(k)},r_{Y\mid T}^{(k)})\}_{k=0}^{\infty}$ is monotonically nonincreasing in $k$.
\end{prop}

\begin{proof}
This follows directly from the exact alternating minimization performed by Algorithm \ref{alg:AO_SDT_IBM} over the blocks $(r_{Y\mid T},q_T,p_{T\mid X})$.
\end{proof}

\begin{prop}
Let us assume that either $\mathcal{H}(Y)$ is nonnegative and upper-bounded, or the corresponding $\text{EVPI}^{\ell_{\mathcal{H}}}(Y)$ is upper-bounded.
Then, sequence $\{G_{\beta}(p_{T\mid X}^{(k+1)},q_T^{(k)},r_{Y\mid T}^{(k)})\}_{k=0}^{\infty}$ converges.
\end{prop}

\begin{proof}
According to Proposition \ref{prop:monotonically_decreasing}, it suffices to show that $G_{\beta}(p_{T\mid X},q_T,r_{Y\mid T})$ is lower-bounded.
Note that
\begin{align}
G_{\beta}(p_{T\mid X},q_T,r_{Y\mid T})
&\ge \min_{q_T}\min_{r_{Y\mid T}} G_{\beta}(p_{T\mid X},q_T,r_{Y\mid T}) \notag\\
&= I(p_X,p_{T\mid X}) - \beta I_{\mathcal{H}}(p_Y,p_{T\mid Y}) \notag\\
&\ge -\beta I_{\mathcal{H}}(p_Y,p_{T\mid Y}).
\end{align}
If $\mathcal{H}(Y)$ is nonnegative and upper-bounded, then
$I_{\mathcal{H}}(p_Y,p_{T\mid Y})\le \mathcal{H}(Y)<\infty$.
If $\text{EVPI}^{\ell_{\mathcal{H}}}(Y)$ is upper-bounded, then, according to Proposition \ref{prop:evsi_leq_evpi},
\begin{align}
I_{\mathcal{H}}(p_Y,p_{T\mid Y})
=\text{EVSI}^{\ell_{\mathcal{H}}}(Y;T)
\le \text{EVPI}^{\ell_{\mathcal{H}}}(Y)
<\infty.
\end{align}
This proves the statement.
\end{proof}

\section{Numerical Examples}\label{sec:numerical_example}

In this section, we present some numerical examples of Algorithm \ref{alg:AO_SDT_IBM}.
Let $\mathcal{X}=\mathcal{Y}=\{1,2,3\}$ and $\mathcal{T}=\{1,2\}$.
We consider the following joint distribution of $(X,Y)$:
\begin{align}
p_{X,Y}=
\begin{bmatrix}
0.02286551 & 0.06322060 & 0.21391389 \\
0.20989825 & 0.03393804 & 0.15616371 \\
0.10464454 & 0.03489356 & 0.16046191
\end{bmatrix},
\end{align}
where the $(i,j)$-th entry corresponds to $p_{X,Y}(i,j)$.

Figure \ref{fig:SDT_IBM} shows the tradeoff curves of the generalized IB problem
for $\beta\in(0,300]$ when adopting the Shannon MI $I(Y;T)$ (blue dots: original IB curve) and
$\mathrm{EVSI}^{\ell_{\mathrm{sq}}}(Y;T)$ (red dots) as the $\mathcal{H}$-MI. 
The stopping threshold was set to $\epsilon=10^{-9}$.
The initial distribution $p_{T\mid X}^{(0)}$ was generated by drawing each entry independently from the continuous uniform distribution on $[0,1]$ and normalizing each row so that it sums to one.
Specifically, we set $p_{T\mid X}^{(0)}$ as follows:
\begin{align}
p_{T\mid X}^{(0)}=
\begin{bmatrix}
0.46707838 & 0.53292162 \\
0.89856339 & 0.10143661 \\
0.45165810 & 0.54834190
\end{bmatrix},
\end{align}
where the $(i,j)$-th entry corresponds to $p_{T\mid X}^{(0)}(j\mid i)$.

In general, the distribution $p_{T\mid X}^{(N)}$ obtained at convergence depends on the choice of the utility measure $I_{\mathcal{H}}(Y;T)$,
where $N$ denotes the number of iterations required to satisfy the stopping criterion.
For example, at $I(X; T)\approx 0.2496$, the resulting distributions are described below.
When using the Shannon MI $I(Y; T)$ as the utility measure, the algorithm converges after $108$ iterations, yielding
\begin{align}
p_{T\mid X}^{(108)}=
\begin{bmatrix}
0.20248578 & 0.79751422 \\
0.95825068 & 0.04174932 \\
0.77937897 & 0.22062103
\end{bmatrix}.
\end{align}
By adopting $\mathrm{EVSI}^{\ell_{\mathrm{sq}}}(Y;T)$ as the utility measure, the algorithm converges after $81$ iterations, yielding 
\begin{align}
p_{T\mid X}^{(81)}=
\begin{bmatrix}
0.33757386 & 0.66242614 \\
0.99267856 & 0.00732144 \\
0.93976497 & 0.06023503
\end{bmatrix}.
\end{align}

\begin{figure}[tb]
\centering
\includegraphics[width=0.5\textwidth]{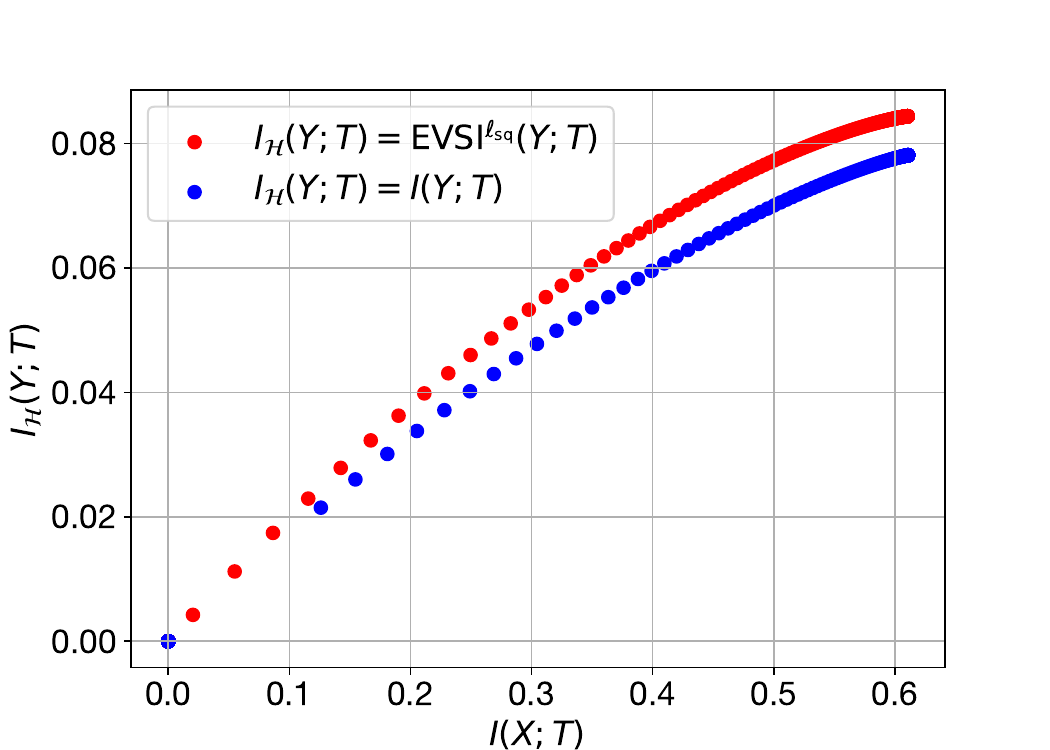}
\caption{
Tradeoff curves of the generalized IB problem when adopting the Shannon MI $I(Y;T)$ (blue dots) and $\mathrm{EVSI}^{\ell_{\mathrm{sq}}}(Y;T)$ (red dots) as utility measures.
}
\label{fig:SDT_IBM}
\end{figure}

\section{Conclusion}\label{sec:conclusion}
In this study, we investigated a generalized IB problem, where the utility of the compressed representation of data is evaluated using an $\mathcal{H}$-MI that satisfies the \texttt{CV} and \texttt{AVG} conditions; the compression measure is the Shannon MI.
By exploiting the decision-theoretic interpretation of this class of $\mathcal{H}$-MI via its EVSI, we derived an alternating optimization algorithm to calculate the tradeoff between compression and utility.
This perspective provides a unified decision-theoretic understanding of generalized IB problems.


\appendices
\section{Proof of Proposition \ref{prop:update_formulae_SDT_IBM}} \label{proof:update_formulae_SDT_IBM}

\begin{proof}
\noindent
$1$. For a fixed $(p_{T\mid X},q_T)$, minimizing $G_{\beta}$ over $r_{Y\mid T}$ is equivalent to maximizing $\mathcal{F}_{\mathcal{H}}(p_{T\mid Y},r_{Y\mid T})$, which is defined in \eqref{eq:vc_H_MI}.
According to Remark \ref{rem:optimal_dist_psr}, $\mathcal{F}_{\mathcal{H}}(p_{T\mid Y}, r_{Y\mid T})$ is maximized when $r_{Y\mid T}(\cdot | t)=p_{Y\mid T}(\cdot | t)$.

\noindent
$2$.  For a fixed $(p_{T\mid X},r_{Y\mid T})$, minimizing $G_{\beta}$ over $q_T$ is equivalent to minimizing $D(p_X p_{T\mid X}\|p_X q_T)$.
Hence, according to Proposition \ref{prop:vc_Shannon_MI}, the minimum is achieved at $q_T=p_T$.

\noindent
$3$. Although the update rule \eqref{eq:update_p_T_X} can be derived from the Karush--Kuhn--Tucker conditions, here we provide a proof using the nonnegativity of the Kullback--Leibler divergence.
\begin{align}
&G_{\beta}(p_{T\mid X},q_T,r_{Y\mid T}) \notag \\ 
&:= D(p_X p_{T\mid X}\|p_X q_T)
-\beta\,\mathcal{F}_{\mathcal{H}}(p_{T\mid Y},r_{Y\mid T}) \notag\\
&= D(p_X p_{T\mid X}\|p_X q_T)
+ \beta\vE_{Y,T}\left[\ell_{\mathcal{H}}\bigl(Y,r_{Y\mid T}(\cdot\mid T)\bigr)\right] \notag \\ 
&\qquad \qquad \qquad \qquad \quad -\beta\,\mathcal{H}(p_Y) \notag\\
&= \vE_X\left[D\left(p_{T\mid X}(\cdot\mid X)\,\|\,p_{T\mid X}^{*}(\cdot\mid X)\right)\right] \notag \\ 
&\qquad + \sum_{x}p_X(x)\log Z_{\beta}(q_T,r_{Y\mid T},x) -\beta\,\mathcal{H}(p_Y) \notag\\
&\overset{(a)}{\geq}
\sum_x p_X(x)\log Z_{\beta}(q_T,r_{Y\mid T},x) -\beta \mathcal{H}(p_Y), \label{eq:lb_p_T_X}
\end{align}
where $p^{*}_{T\mid X}$ is defined in \eqref{eq:update_p_T_X}, 
\begin{align}
&Z_{\beta}(q_T,r_{Y\mid T},x) \notag \\ 
&:=\sum_t q_T(t)\exp\left\{-\beta \sum_y p_{Y\mid X}(y\mid x)\,
\ell_{\mathcal{H}}\bigl(y,r_{Y\mid T}(\cdot\mid t)\bigr)\right\},
\end{align}
and $(a)$ follows from the nonnegativity of the Kullback--Leibler divergence.
The equality in \eqref{eq:lb_p_T_X} holds if and only if $p_{T\mid X}=p_{T\mid X}^{*}$ (see, e.g., \cite[Thm. 2.6.3]{Cover:2006:EIT:1146355}).
\end{proof}

\end{document}